\documentclass[10pt, conference]{IEEEtran}
\IEEEoverridecommandlockouts
\usepackage{amsmath,amsthm,dsfont,tabularx,cuted}
\usepackage[cmintegrals]{newtxmath}
\usepackage[nocompress]{cite}
\usepackage[utf8]{inputenc}
\usepackage[english]{babel}
\usepackage[dvips]{graphicx}
\usepackage{subcaption}
\usepackage{algorithm,algorithmic}
\usepackage{caption}
\usepackage{bm}
\usepackage[noblocks]{authblk}
\usepackage{color}
\usepackage{multirow}
\usepackage{siunitx}
\usepackage{setspace}
\usepackage[dvipsnames]{xcolor}
\usepackage{comment,soul}
\usepackage{bbm}

\DeclareMathOperator*{\argmax}{arg\,max}

\floatname{Procedure}{algorithm}

\newtheorem{proposition}{Proposition}

\newcommand{\ket}[1]{|#1\rangle}
\newcommand{\bra}[1]{\langle #1|}

\begin{document}
\title{Qubit Recycling in Entanglement Distillation \\
\thanks{This work was supported in part by the National Science Foundation under grant OMA-2304118.}
}

\author{Stuart Pelletier, Ruozhou Yu, George Rouskas, Jianqing Liu \\
Department of Computer Science, North Carolina State University, Raleigh, NC 27606, USA \\
E-mail: \{sopellet, ryu5, rouskas, jliu96\}@ncsu.edu
}

\maketitle

\begin{abstract}
Quantum entanglement distillation is a process to extract a small number of high-fidelity entanglement from a large number of low-fidelity ones, which in essence is to trade yield (or survival rate) for fidelity. Among existing distillation approaches, Gisin's local filtering protocol is commonly adopted in photonic quantum systems for distilling entangled photons in polarization basis. Yet, the performance of Gisin's filter is cursed by the same fundamental trade-off between fidelity and yield. To address this challenge, in this work, we propose a protocol to recycle the disposed photons and improve their fidelity by a designed (and optimized) local operator. The key parameters of the proposed protocol are calculated by solving a constrained optimization problem. In so doing, we achieve significantly higher yield of high-fidelity entanglement pairs. We further evaluate the performance of our designed protocol under two common configurations of Gisin's filter, namely full filter and partial filter. Compared with existing distillation protocols, the results demonstrate that our design achieves as much as 31.2\% gain in yield under the same fidelity, while only incurring moderate system complexity in terms of  invested hardware and extra signaling for synchronization.
\end{abstract}


\section{Introduction}
Quantum entanglement as a physical phenomenon in the microscopic world once troubled Einstein who called it ``spooky action at a distance,'' but it was later validated by the well-known Bell inequality test. Nowadays, despite many unanswered scientific questions around quantum entanglement, quantum networks have been widely engineered and deployed around the globe. The common goal of all these quantum networks is to distribute entanglement in large volume and high quality \cite{chen2022heuristic}, as entanglement is central to numerous applications in future quantum internet such as quantum teleportation, quantum computation, and quantum cryptography \cite{yu2022topology,li2022cluster}.

When interacting with the environment like quantum memory and fibre channels, quantum entanglement unavoidably experiences coherence degradation that may lead to entanglement sudden death \cite{almeida2007environment}. The common way to cope with decoherence is entanglement distillation, by which a smaller number of highly entangled states are extracted from a large number of weakly entangled states \cite{bennett1996concentrating}. Among existing entanglement distillation protocols, Bennett's controlled-NOT (CNOT) operation \cite{bennett1996mixed} and Gisin's local filtering operation \cite{gisin1996hidden} are featured as mainstream approaches. Compared with Bennett's approach, Gisin's local filter has two appealing merits: (1) only local operations are needed (i.e., no classical communications); (2) only a single copy of the entangled state is needed (i.e., no ancilla entanglements are scarified). 

Since its inception in 1996, Gisin's local filter has been extensively researched in both theory and experiments for entanglement distillation. In principle, a pair of weakly entangled qubits (and likewise for multipartite ($>$2 qubits) entanglement, such as the GHZ state) can become strongly entangled when passing through Gisin's filters. Any qubits reflected by the filter, however, will have their entanglement weakened, or in some cases, destroyed. Such qubits can either be measured or discarded as they are deemed useless at that point. While this uselessness holds true in many (ideal) cases, for some input states and/or under certain (practical) filter configurations, these reflected qubits are shown to have non-zero concurrence, i.e., they are still entangled despite weak strength. A natural question to ask is whether such reflected qubits can be recycled and turned into strongly entangled states. One can obviously anticipate a much higher yield of usable entanglement if the answer to this question is affirmative. 

To this end, we present in this paper a novel protocol --- consisting of a non-unitary transformation and multi-party agreement on coincidence count --- to harvest and improve the weakly entangled qubits that are reflected by Gisin's filters. To search for the optimal non-unitary operator, we formulate a constrained optimization problem that maximizes the high-fidelity survival rate, i.e., the total entanglement yield with the minimum requirement on their fidelity. The protocol is integrated into and examined under two common filter-based entanglement distillation setups, namely the full filtering and partial filtering schemes. Based on numerical simulations, we demonstrate the superior performance of our qubit-recycling protocol in terms of high-fidelity survival rate compared to existing filter schemes.

The paper is organized as follows. To begin with, we survey the recent advances in entanglement distillation in Section \ref{related}. Next, we introduce the basic concepts that are relevant to our research problem in Section \ref{preliminary}. We then describe the principle and design details of our proposed protocol in Section \ref{recycle}. To evaluate the performance of the protocol, we present the simulation results in Section \ref{performance}. Lastly, we conclude the paper in Section \ref{conclusion} with an outlook for the future work.

\section{Related Works}\label{related}
In this section, we review recent advancements in entanglement distillation that have contributed to the ongoing development of the field. We organize our discussion into three subtopics: (1) distillation of multipartite states, which extends the scope of entanglement distillation beyond simple bipartite systems; (2) distillation using hyperentanglement, an emerging approach that utilizes multiple modes of entanglement to enhance  distillation; and (3) distillation using reset-and-reuse operations in a quantum computer, a novel methodology that employs the inherent capabilities of quantum computing hardware to facilitate the distillation process by recycling and re-entangling ancilla qubits. By examining these recent developments, we aim to provide an overview of the current state of entanglement distillation research and highlight the significance and novelty of our proposed qubit recycling protocol.

\subsubsection{Distillation of Multipartite Entanglement States}

The distillation of multipartite entangled states, such as GHZ states, has garnered attention due to the advantages of entanglement being shared between more than two parties. Huang et al.~\cite{huang2014distillation} proposed a single-copy-based distillation scheme for amplitude-damped W states and amplitude-damped GHZ states. De Bone et al.~\cite{de2020protocols} investigated the creation and distillation of GHZ states out of nonperfect Bell pairs. They introduced a heuristic dynamic programming algorithm to optimize protocols for creating and purifying GHZ states.

\subsubsection{Distillation Utilizing Hyperentanglement}

Utilizing hyperentanglement has been explored as a promising technique for enhancing entanglement distillation schemes. Zhou and Sheng \cite{zhou2021high} proposed an efficient two-step entanglement purification protocol for polarization entanglement using a single copy of states by utilizing hyperentanglement in the time bin and spatial modes. Ecker et al. \cite{ecker2021experimental} experimentally demonstrated single-copy entanglement distillation using pairs of single photons entangled in both the polarization and energy-time domains. 

\subsubsection{Reset-and-Reuse}

In recent work by Germain et al. \cite{germain2022quantum}, the authors explore the potential of a reset-and-reuse operation in quantum computers to substantially reduce yield loss in entanglement distillation protocols. They implement multi-pass distillation schemes, specifically BBPSSW and DEJMPS, and test them on the IBM-Q environment. This reset-and-reuse feature shows a significant minimization in the number of qubits required for distillation, bringing the number of qubits required per pass down from exponential to constant --- a notably large improvement. It should be noted that such a reset-and-reuse operation, while available in quantum computers, is not currently available in a quantum network setting, as there are many challenges associated with re-entangling distance-separated ancillary qubits after measurement. Our work proposes a novel single-copy qubit recycling protocol which does not require any such re-entangling and can thus be used by a quantum network with currently available hardware.

\section{Preliminaries}\label{preliminary}
\subsection{Gisin's Local Filter}
In the demonstrative experiment by Kwiat \cite{kwiat2001experimental}, Gisin's local filter was realized by a series of coated glass slabs, tilted against the vertical axis by the Brewster's angle, as shown by an example in Fig.~\ref{fig:bp}. By adjusting the configuration of these slabs (e.g., angles and coated materials), the transmission probability $T_{H}$ (resp., $T_{V}$) for horizontally (resp. vertically) polarized incident photons can be tuned, owing to the well-known polarization-dependent reflectivity \cite{born2013principles}. As a result, undesired states (i.e., noises) can be selectively blocked (and reflected in another direction), thus leaving the surviving photons to be more concentrated in the desired entangled states. 
\begin{figure}
\centering
  \includegraphics[width=0.65\linewidth]{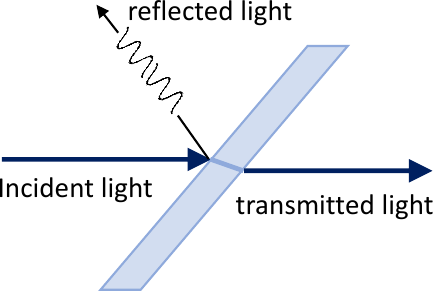}
  \caption{Gisin's filter implemented by a Brewster plate.}
  \label{fig:bp}
\end{figure}
In theory, the Gisin's local filter can be modeled as a positive operator-valued measurement (POVM), namely $\{M_0, M_1\}$ where $M_0 = \big(\begin{smallmatrix}
  \alpha & 0\\
  0 & \beta
\end{smallmatrix}\big)$ and $M_1 = I - M_0$ are positive semi-definite Hermitian. $M_0$ ($M_1$ likewise) is realized by the projector $m_0 = \sqrt{\alpha}\ket{0}\bra{0}+\sqrt{\beta}\ket{1}\bra{1}$ and $M_0 = m_0 * m_0^{\dagger}$. When implementing the POVM (or Gisin's filter) in  photonic systems, $\alpha$ and $\beta$ respectively denote the transmission probability $T_{H}$ and $T_{V}$ of the glass slabs. That is to say, the design of Gisin's local filter is boiled down to the construction of $\alpha'$s and $\beta'$s. 

\subsection{Channel Decoherence Model}
In this work, we consider a (photonic) quantum network that distributes EPR pairs between any two arbitrary nodes. An entanglement source (ES) generates EPR pairs by directing a laser beam at a BBO (beta-barium borate) crystal. Without loss of generality, the EPR pair in the state of $\ket{\Phi^+}$ = $\frac{1}{\sqrt{2}}(\ket{00}+\ket{11})$ or $\rho$ = $\ket{\Phi^+}\bra{\Phi^+}$ is assumed. 

Then, each qubit of the EPR pair is distributed to Alice and Bob through independent decoherence channels. We consider the amplitude damping model in which state $\ket{1}$ may decay into $\ket{0}$. Mathematically, an amplitude damping channel $\mathcal{E}$ is described by the following super-operators, a.k.a, Kraus operators:
\begin{equation}\label{Damping_Kraus_Operators}
    E^{i}_{0} = \begin{bmatrix}
        1&0\\
        0&\sqrt{\bar{\gamma_i}}
    \end{bmatrix}, \quad E^{i}_{1} = \begin{bmatrix}
        0&\sqrt{\gamma_i}\\
        0&0
    \end{bmatrix},
\end{equation}
where $i \in \{A,B\}$, $\gamma_i  = 1 - e^{-{t_i}/{T_1}}$ is a time-dependent damping factor in which $T_1$ is defined as the time it takes for the $\ket{1}$ state to settle into the $\ket{0}$ (vice versa). Denote $\bar{\gamma_i} = 1 - \gamma_i$. After channel decoherence, the received state at Alice and Bob is 
\begin{equation}\label{state_after_channel}
    \rho^{\prime} = \mathcal{E}(\rho)=\sum^{1}_{j=0} \sum^{1}_{k=0} \left(E_j^A \otimes E_k^B\right) \rho \left(E_j^{A} \otimes E_k^{B}\right)^{\dagger}.
\end{equation}

For the sake of notation simplicity, in the remainder of this paper, we consider the same fading channel for ES-A and ES-B, i.e., $\gamma = \gamma_A = \gamma_B$.

\section{Design Principles of Qubit Recycling}\label{recycle}
In this section, we consider two common entanglement distillation setups in the literature, with one being that both Alice and Bob implement Gisin's local filters (coined as ``full filtering'') while the other being that either Alice or Bob implements a Gisin's local filter (coined as ``partial filtering''). While both setups have their merits, we will investigate the best use case of our proposed qubit-recycling idea and how much gain it can offer. 
\subsection{Qubit Recycling under Full Filtering}
\subsubsection{Typical full filtering design} To offset the decoherence incurred by the amplitude damping channel and restore the received state $\rho^{\prime}$ closer to its original entanglement state $\rho$, Alice and Bob implement Gisin's local filters, which are mathematically defined as the POVMs $\{M_{A,0}, M_{A,1}\}$ and $\{M_{B,0}, M_{B,1}\}$ respectively, for entanglement distillation. We consider the local filters performed by Alice and Bob described by the operation: 
\begin{equation}\label{POVM}
    M_{i,0} = \begin{bmatrix}
        \alpha_i & 0\\
        0 & \beta_i 
    \end{bmatrix},M_{i,1} = \begin{bmatrix}
        \beta_i  & 0\\
        0 & \alpha_i
    \end{bmatrix},
\end{equation}
where $\alpha_i,\beta_i \in (0,1)$ and $\alpha_1 + \beta_i = 1$ complying with the POVM's property. In existing work, full filtering schemes have been widely explored, wherein Alice and Bob each distills her/his respective qubit independently. This process is mathematically described by applying POVMs on both qubits. We refer to the state after undergoing both filters, i.e., the state Alice and Bob want to keep, as 
\begin{equation}\label{Transmitted_Full_Filtered_State}
    \Tilde{\rho}_{11} = \frac{1}{S_{11}}(\sqrt{M_{A,1}} \otimes \sqrt{M_{B,1}}) \rho^{\prime} (\sqrt{M_{A,1}} \otimes \sqrt{M_{B,1}})^{\dagger}.
\end{equation}
where $S_{11}$ is the normalization factor that is $S_{11} = \text{Tr}\{(\sqrt{M_{A,1}} \otimes \sqrt{M_{B,1}}) \rho^{\prime} (\sqrt{M_{A,1}} \otimes \sqrt{M_{B,1}})^{\dagger}\}$. The value of $S$ represents the likelihood that both Alice's and Bob's qubits pass through the Gisin's local filters, thus can be considered as the success probability, or survival rate, of the distillation process. Note that as we consider indentical channels for ES-A and ES-B, that is $\gamma_A = \gamma_B$, Alice's and Bob's filter will have the same configurations. Therefore, we can drop the subscript for A and B and simply let $\alpha = \alpha_A = \alpha_B$ (likewise for $\beta$). 

The calculation of the POVM parameters $\{\alpha, \beta\}$ is usually performed by solving a constrained optimization problem that seeks to maximize the high-fidelity yield, i.e. the success probability while meeting a minimum requirement on the entanglement fidelity. The reason for posing a hard constraint on fidelity is because some quantum applications (e.g., QKD) have a stringent requirement on the minimum fidelity to be considered usable (e.g., satisfying a minimum secret key rate) \cite{zhang2022device}. Mathematically,
\begin{equation}\label{fidelity_optimization_full}
\{\alpha^*, \beta^*\} = \argmax_{\{\alpha, \beta\}} \, S_{11}; \quad \text{s.t.} \, \text{Tr}\left[\sqrt{\sqrt{\rho}\Tilde{\rho}_{11}\sqrt{\rho}}\right]^2 \geq  F_{th}.
\end{equation}
This problem is a typical multivariate quadratic optimization problem, which can be easily proven to be convex by checking the second order derivatives of the objective and constraint functions. By Slater’s condition, the necessary and sufficient conditions for a solution $\{\alpha^*, \beta^* \}$ to be the optimal solution are the KKT conditions. 

\subsubsection{Residue entanglement in reflected qubits} When Alice's and Bob's local filters are configured using the parameters $\{\alpha^*, \beta^* \}$, a photon pair passing through both filters is guaranteed to have a desired fidelity level. Yet, with such optimal filter configuration, the reflected qubit(s) could still be usable in the sense that there is a certain degree of entanglement remained. 

\begin{proposition}
Suppose an EPR pair passes through an amplitude damping channel with parameter $\gamma$ and is filtered using Gisin's local filter with a POVM with parameters $\{\alpha,\beta\}$. The resulting state of the reflected photons, $\tilde{\rho}_{00}$, is entangled when $\alpha,\beta \neq 0$ and $\gamma \neq 1$.
\end{proposition}
\begin{proof}
Note that $\tilde{\rho}_{00} = $
\begin{align*}\label{RR_Density_Matrix}
\begin{bmatrix}
\alpha^2\left(\frac{1}{2}+\frac{\gamma^2}{2}\right) & 0 & 0 & \frac{1}{2}\alpha\beta\left(1-\gamma\right) \\
0 & \frac{1}{2}\alpha\beta\left(1-\gamma\right)\gamma & 0 & 0 \\
0 & 0 & \frac{1}{2}\alpha\beta\left(1-\gamma\right)\gamma & 0 \\
\frac{1}{2}\alpha\beta\left(1-\gamma\right) & 0 & 0 & \frac{1}{2}\beta^2\left(1-\gamma\right)^2
\end{bmatrix}.
\end{align*}
This density matrix is separable if and only if its partial transpose is positive \cite{doherty2004complete}. This is called the PPT condition, which is equivalent to the condition that its partial transpose has exclusively non-negative eigenvalues. In other words, if at least one of its eigenvalues is negative, then the state $\tilde{\rho}_{00}$ is entangled. Note that its partial transpose\footnote{The partial transpose generally is taken with respect to one qubit, corresponding to either Alice's or Bob's qubit. However, the eigenvalues of the partial transpose are invariant under which qubit the partial transpose is taken on, because the partial transpose with respect to Alice's qubit is equal to the transpose of the partial transpose taken with respect to Bob's qubit. In this case, then, since the partial transpose is symmetric, it is the same partial transpose matrix for both Alice's and Bob's qubits.} is the density matrix
\begin{equation*}\label{RR_PPT}
\begin{bmatrix}
\alpha^2 \left(\frac{1}{2} + \frac{\gamma^2}{2}\right) & 0 & 0 & 0 \\
0 & \frac{1}{2}\alpha\beta(1-\gamma)\gamma & \frac{1}{2}\alpha\beta(1-\gamma) & 0 \\
0 & \frac{1}{2}\alpha\beta(1-\gamma) & \frac{1}{2}\alpha\beta(1-\gamma)\gamma & 0 \\
0 & 0 & 0 & \frac{1}{2}\beta^2(1-\gamma)^2
\end{bmatrix}
\end{equation*}
which has eigenvalues
\begin{equation}
\begin{aligned}
\lambda_1 &= -\frac{1}{2}\alpha\beta(-1+\gamma)^2, \quad \lambda_2 = \frac{1}{2}\beta^2(-1+\gamma)^2, \\
\lambda_3 &= \frac{1}{2}\alpha^2(1+\gamma^2), \qquad \,\,\, \lambda_4 = \frac{1}{2}\alpha\beta(1 - \gamma^2).
\end{aligned}
\end{equation}
Note that $\lambda_2,\lambda_3$ and $\lambda_4$ all take on non-negative values for all $\alpha,\beta,\gamma \in [0,1]$. The eigenvalue $\lambda_1$, however, takes on a negative value except when $\alpha=0$, $\beta=0$ or $\gamma=1$. Therefore, $\tilde{\rho}_{00}$, is entangled when $\alpha,\beta \neq 0$ and $\gamma \neq 1$.
\end{proof}
\begin{figure}
\centering
  \includegraphics[width=\linewidth]{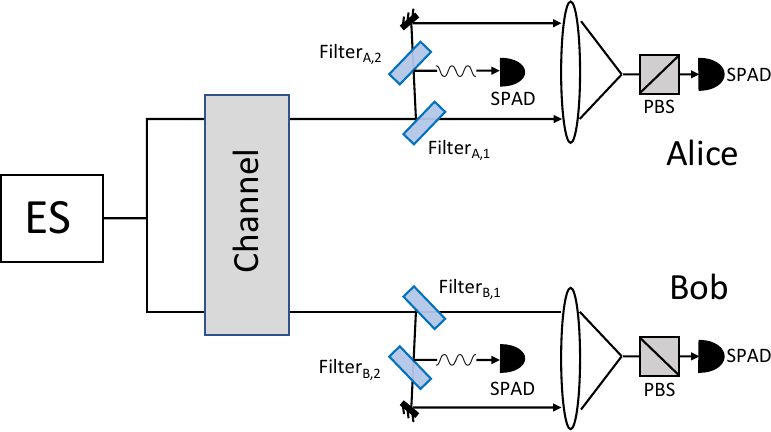}
  \caption{Qubit recycling under the full filtering setup.}
  \label{fig:full_filter}
\end{figure}

\subsubsection{Recycling reflected qubits} In light of the remaining usable entanglement in the reflected qubits, we propose a second Gisin's local filter, denoted as $\texttt{Filter}_\texttt{{A/B,2}}$, to harvest them. The basic idea is shown in Fig. \ref{fig:full_filter}, in which the reflected qubits are distilled by another filter. Then, the two light paths are integrated and analyzed by a single-photon avalanche detector (SPAD). Note that a small portion of the reflected qubits from $\texttt{Filter}_\texttt{{*,1}}$ will be reflected by $\texttt{Filter}_\texttt{{*,2}}$ again. While they can be looped back for further recycling, we choose to measure them as their entanglement strength becomes much weaker than that observed when they are only reflected once. Technically, by calculating the concurrence following Proposition 1, we can show that the entanglement strength progressively deteriorates as qubits are reflected by each subsequent filter. 

To determine the optimal configurations of $\texttt{Filter}_\texttt{{*,2}}$, let us first define an outcome space for $\texttt{Filter}_\texttt{{*,1}}$ as $\Omega_1$ = $\{T_{A,1}T_{B,1}, \, T_{A,1}R_{B,1}, \, R_{A,1}T_{B,1}, \, R_{A,1}R_{B,1} \}$. For example, the outcome $w = R_{A,1}T_{B,1}$ implies that Alice's qubit is reflected while Bob's is transmitted. In the traditional full filtering scheme, this outcome would be considered a failure because no coincidence click is observed. In addition, we can define the outcome space for the second-tier local filters $\Omega_2$ = $\{$ $\emptyset_{A,2}\emptyset_{B,2}$, $T_{A,2}\emptyset_{B,2}$,  $R_{A,2}\emptyset_{B,2}$, $\emptyset_{A,2}T_{B,2}$, $\emptyset_{A,2}R_{B,2}$, $T_{A,2}T_{B,2}$, $T_{A,2}R_{B,2}$, $R_{A,2}T_{B,2}$, $R_{A,2}R_{B,2} \}$ in which $\emptyset$ is an null event that implicitly tells that no qubit arrives at this filter. Among these possible outcomes, we collect the outcomes which result in the final distilled entanglement in a set $\Omega_\checkmark$ = $\{T_{A,1}T_{B,1} \land \emptyset_{A,2}\emptyset_{B,2}$, $T_{A,1}R_{B,1} \land \emptyset_{A,2}T_{B,2}$, $R_{A,1}T_{B,1} \land T_{A,2}\emptyset_{B,2}$, $R_{A,1}R_{B,1} \land T_{A,2}T_{B,2}\}$ which gives us the survival rate $P_\checkmark = \sum_{i=1}^{4} \text{Pr}(\omega_i \in \Omega_\checkmark)$.

Specifically, the survival rates for the four cases in $\Omega_\checkmark$ are respectively calculated as follows
\begin{equation*}
\text{Pr}(T_{A,1}T_{B,1} \land \emptyset_{A,2}\emptyset_{B,2}) = S_{11}
\end{equation*}
\begin{equation*}
\begin{aligned}
    \text{Pr}(T_{A,1}R_{B,1} &\land \emptyset_{A,2}T_{B,2}) = \\
    &\text{Tr}\{(\sqrt{M_{A,1}} \otimes \sqrt{M_{B,0}}) \rho^{\prime} (\sqrt{M_{A,1}} \otimes \sqrt{M_{B,0}})^{\dagger}\} \\
    &\times \text{Tr}\{(I \otimes \sqrt{M_{B,1}^{'}}) \Tilde{\rho}_{10} (I \otimes \sqrt{M_{B,1}^{'}})^{\dagger}\}
\end{aligned}
\end{equation*}
\begin{equation*}
\begin{aligned}
    \text{Pr}(R_{A,1}T_{B,1} &\land T_{A,2}\emptyset_{B,2}) = \\
    &\text{Tr}\{(\sqrt{M_{A,0}} \otimes \sqrt{M_{B,1}}) \rho^{\prime} (\sqrt{M_{A,0}} \otimes \sqrt{M_{B,1}})^{\dagger}\} \\
    &\times \text{Tr}\{(\sqrt{M_{A,1}^{'}} \otimes I) \Tilde{\rho}_{01} (\sqrt{M_{A,1}^{'}} \otimes I)^{\dagger}\}
\end{aligned}
\end{equation*}
\begin{equation*}
\begin{aligned}
    \text{Pr}(R_{A,1}R_{B,1} &\land T_{A,2}T_{B,2}) = \\
    &\text{Tr}\{(\sqrt{M_{A,0}} \otimes \sqrt{M_{B,0}}) \rho^{\prime} (\sqrt{M_{A,0}} \otimes \sqrt{M_{B,0}})^{\dagger}\} \\
    &\times \text{Tr}\{(\sqrt{M_{A,1}^{'}} \otimes \sqrt{M_{B,1}^{'}}) \Tilde{\rho}_{00} (\sqrt{M_{A,1}^{'}} \otimes \sqrt{M_{B,1}^{'}})^{\dagger}\}
\end{aligned}
\end{equation*}
where the second-tier filter's POVM operator is captured by $\{M_{A/B,0}^{'}, M_{A/B,1}^{'}\}$. Moreover, for any cases in $\omega_i \in \Omega_\checkmark$, we denote the output quantum state as $\hat{\rho}_{11, \omega_i}$ which can be calulated similar to Eq.~(\ref{Transmitted_Full_Filtered_State}).

Then, the search of optimal $\{\alpha', \beta'\}$ for the POVM operator $\{M_{A/B,0}^{'}, M_{A/B,1}^{'}\}$ of $\texttt{Filter}_\texttt{{*,2}}$ is formulated as the following optimization problem.
\begin{equation}\label{recycling_optimization_full}
\begin{aligned}
&\{\alpha'^*, \beta'^*\} = \\ 
&\argmax_{\{\alpha', \beta'\}} \sum_{i=1}^{4}\, \text{Pr}(\omega_i \in \Omega_\checkmark) \, \cdot \, \mathbbm{1}(\text{Tr}\left[\sqrt{\sqrt{\rho}\hat{\rho}_{11, \omega_i}\sqrt{\rho}}\right]^2 \geq  F_{th}),
\end{aligned}
\end{equation}
in which $\mathbbm{1}(\cdot)$ is the indicator function that is 1 if its provided statement is true; and 0 otherwise.

\subsection{Qubit Recycling Under Partial Filtering}
Partial filtering is another widely adopted configuration in entanglement distillation for its higher survival rate. In its setup, depending on which channel has stronger decoherence, only one of Alice or Bob implements a local filter. This setup naturally gives rise to a higher survival rate without losing too much of the fidelity. Since this paper considers identical channel decoherence on ES-A and ES-B, there is no difference of placing a filter on Alice's or Bob's end. Therefore, without loss of generality, we consider the setup  in which Alice filters her qubit, while Bob does not. 

First of all, examining the single-filter case, we call the state transmitted by $\texttt{Filter}_\texttt{{A,1}}$, i.e., the state Alice and Bob want to keep in a traditional partial filtering design without qubit recycling, as 
\begin{equation}\label{Transmitted_Partial_Filtered_State}
    \Tilde{\rho}_{1} = \frac{1}{S_{1}}(\sqrt{M_{A,1}} \otimes I) \rho^{\prime} (\sqrt{M_{A,1}} \otimes I)^{\dagger},
\end{equation} where $S_1$ is a the normalization factor. The goal is find the optimized parameters for $\texttt{Filter}_\texttt{{A,1}}$ by solving a fidelity-constrained yield-maximization problem similar to (\ref{fidelity_optimization_full}). Mathematically,
\begin{equation}\label{fidelity_optimization_partial}
\{\alpha^*, \beta^*\} = \argmax_{\{\alpha, \beta\}} \, S_{1}; \quad \text{s.t.} \, \text{Tr}\left[\sqrt{\sqrt{\rho}\Tilde{\rho}_{1}\sqrt{\rho}}\right]^2 \geq  F_{th}.
\end{equation}

Moreover, we define the outcome space of $\texttt{Filter}_\texttt{{A,1}}$ as $\Omega'_1$ = $\{T_{A,1}, R_{A,1}\}$, and that of $\texttt{Filter}_\texttt{{A,2}}$ as $\Omega'_2$ = $\{ \emptyset_{A,2}, T_{A,2}, R_{A,2}\}$. Analogously to the full filter case, we collect the outcomes which result in a distilled entanglement pair, giving us the set $\Omega'_\checkmark$ = $\{T_{A,1} \land \emptyset_{A,2}, R_{A_1} \land T_{A,2}\}$ and the survival rate $P'_\checkmark = \sum_{i=1}^{2} \text{Pr}(\omega_i \in \Omega'_\checkmark)$. This similarly leads us to the following analogous constrained optimization problem.
\begin{equation}\label{recycling_optimization_partial}
\begin{aligned}
&\{\alpha'^{*}, \beta'^{*}\} = \\ 
&\argmax_{\{\alpha', \beta'\}} \sum_{i=1}^{2}\, \text{Pr}(\omega'_i \in \Omega'_\checkmark) \, \cdot \, \mathbbm{1}(\text{Tr}\left[\sqrt{\sqrt{\rho}\hat{\rho}_{1, \omega'_i}\sqrt{\rho}}\right]^2 \geq  F_{th}),
\end{aligned}
\end{equation}.
\section{Performance Evaluation}\label{performance}
\begin{figure*}[!t]
\centering
\captionsetup{justification=centering}
\begin{subfigure}[t]{0.33\textwidth}
  \includegraphics[width=\linewidth]{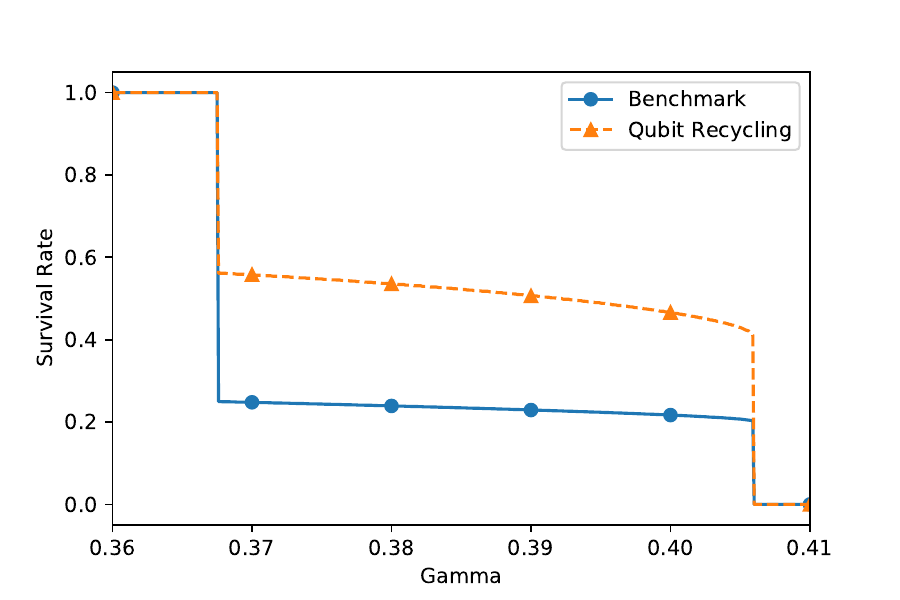}
  \caption{\footnotesize Full Filter with $F_{th} = 0.7$} \label{full_0.7}
\end{subfigure}\hfill
\begin{subfigure}[t]{0.33\textwidth}
  \includegraphics[width=\linewidth]{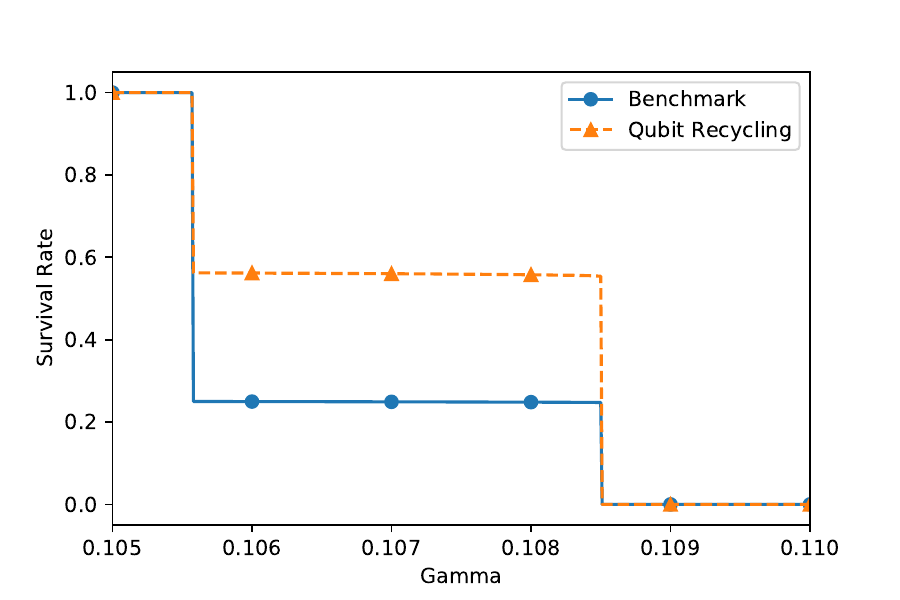}
  \caption{\footnotesize Full Filter with $F_{th} = 0.9$} \label{full_0.9}
\end{subfigure}\hfill
\begin{subfigure}[t]{0.33\textwidth}
  \includegraphics[width=\linewidth]{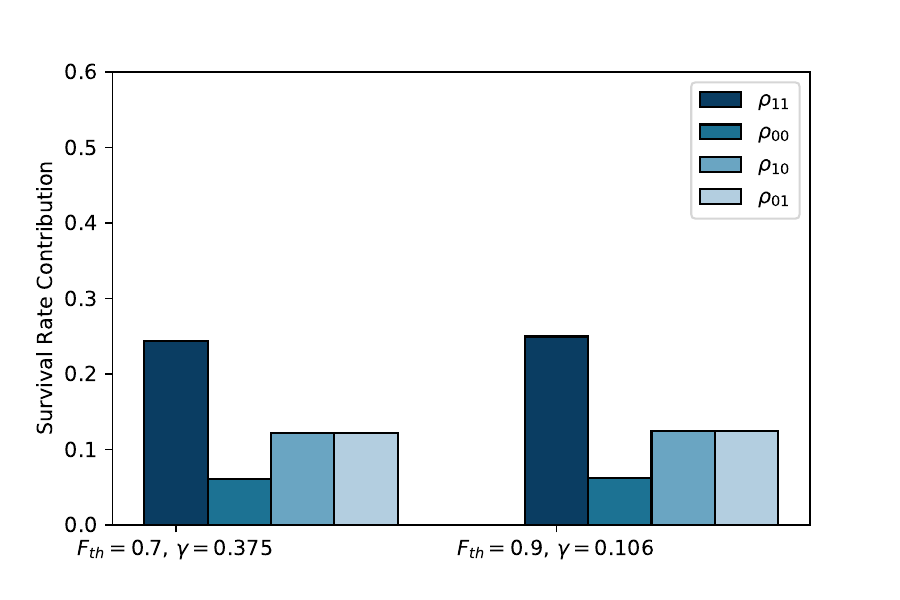}
  \caption{\footnotesize Full Filter outcome breakdown} \label{full_histogram}
\end{subfigure}\hfill
\\
\vspace{0.05in}
\begin{subfigure}[t]{0.33\textwidth}
  \includegraphics[width=\linewidth]{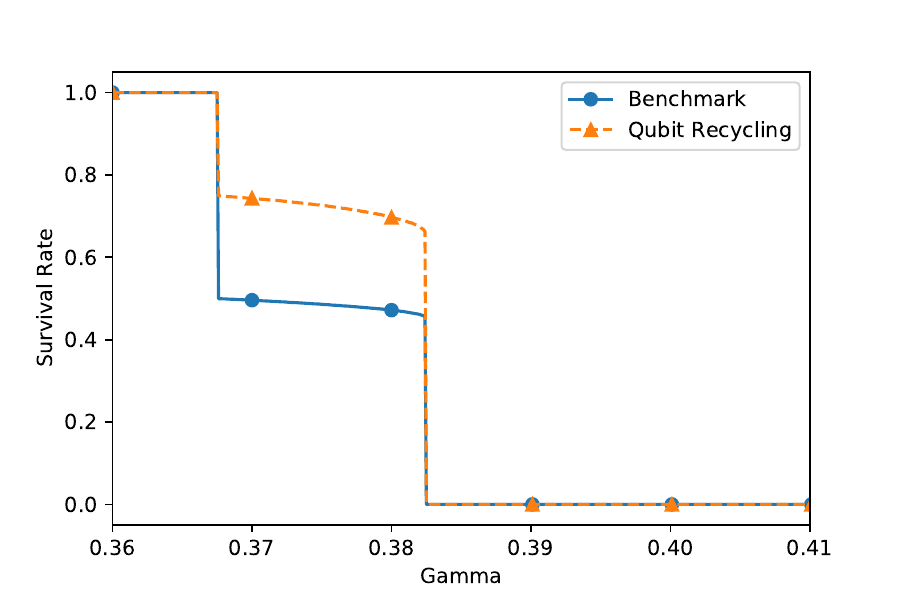}
  \caption{\footnotesize Partial Filter with $F_{th} = 0.7$} \label{partial_0.7}
\end{subfigure}\hfill
\begin{subfigure}[t]{0.33\textwidth}
  \includegraphics[width=\linewidth]{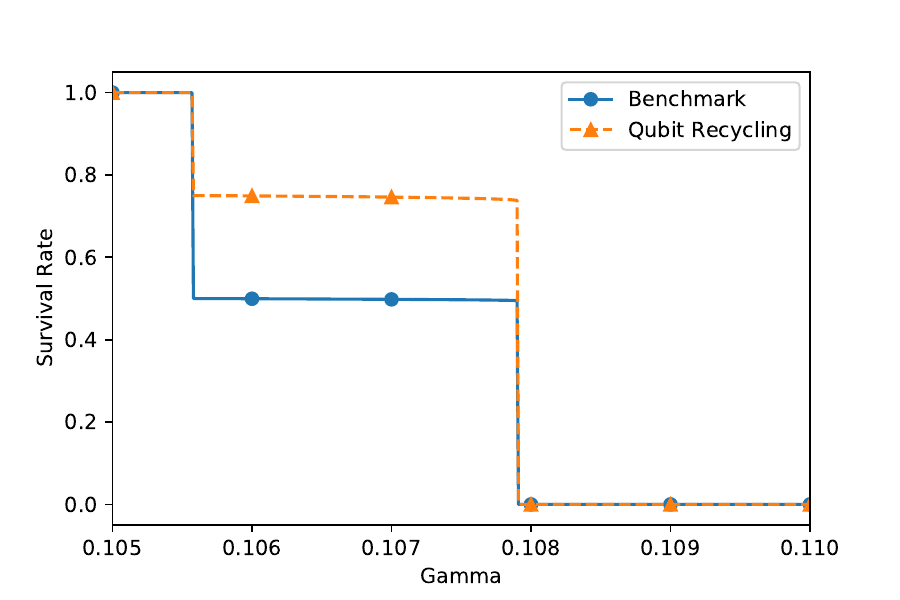}
  \caption{\footnotesize Partial Filter with $F_{th} = 0.9$} \label{partial_0.9}
\end{subfigure}\hfill
\begin{subfigure}[t]{0.33\textwidth}
  \includegraphics[width=\linewidth]{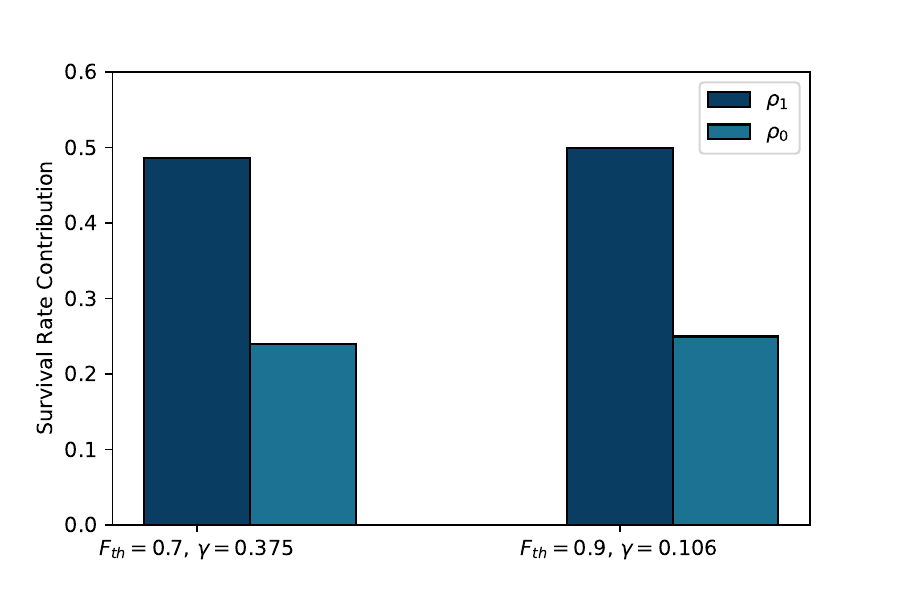}
  \caption{\footnotesize Partial Filter outcome breakdown} \label{partial_histogram}
\end{subfigure}\hfill
   \caption{\small (a, b) The survival rates with respect to $\gamma$ for given $F_{th}$ values for the full filtering scheme and (d, e) for the partial filtering scheme. A breakdown of which outcomes contribute to survival rate for a given $\gamma$ and $F_{th}$ is plotted for (c) the full filtering and (f) the partial filtering schemes.} \label{results}
\end{figure*}

\subsection{Simulation Methodology}
In order to evaluate the performance of our proposed qubit recycling protocol, we developed a simulation model which solves the constrained optimization problems (\ref{fidelity_optimization_full}), (\ref{recycling_optimization_full}), (\ref{fidelity_optimization_partial}), and (\ref{recycling_optimization_partial}). The simulation is implemented in Python, and consists of the following steps:

\begin{enumerate}
  \item \textbf{Initialization}: At the beginning of the simulation, the initial parameters and constraints of the problem are defined. The quantum system $\rho$ is prepared, we define a range of $\gamma$ values to evaluate, and we fix our $F_{th}$ value. Specifically, $F_{th}$ values of 0.7 and 0.9 were selected.
  
  \item \textbf{First filter parameter optimization}: The simulation first assumes a single filter model as a benchmark, and refines the parameters of the local POVM operator $\texttt{Filter}_\texttt{{*,1}}$ through an iterative optimization algorithm. The optimization process iterates through the given $\gamma$ value range for our given $F_{th}$ value and finds the $\{\alpha, \beta\}$ values which respectively maximize (\ref{fidelity_optimization_full}) and (\ref{fidelity_optimization_partial}). 

  \item \textbf{Second filter parameter optimization}: Given the optimized $\{\alpha, \beta\}$ value corresponding to a given $\gamma$ and $F_{th}$ for $\texttt{Filter}_\texttt{{*,1}}$, a second filter $\texttt{Filter}_\texttt{{*,2}}$ is optimized using similar iterative methods to solve (\ref{recycling_optimization_full}) and (\ref{recycling_optimization_partial}).
  
  \item \textbf{Evaluation}: The optimized local operators are then applied to the prepared quantum system, and the survival rate and fidelity of the resulting entanglement pairs are calculated, both for the normal filtering case (i.e., benchmark), and for filtering with recycling case, for comparison. Specifically, the normal filtering case is separately instantiated with full filtering and partial filtering schemes. 
  
\end{enumerate}

By following the aforementioned simulation methodology, we are able to determine the optimal design of our local operator for recycling the disposed photons, achieving a significant increase in high-fidelity survival rate over the optimized benchmark scheme. In the following subsections, we will discuss the specific results obtained for the full filtering and partial filtering schemes.

\subsection{Full Filter Results}

Our simulation results demonstrate that the full filtering scheme with qubit recycling shows a significant improvement in survival rate compared to the benchmark single filter protocol, shown in Fig. \ref{full_0.7} and Fig. \ref{full_0.9}. For the $F_{th} = 0.7$ case, our design adds $20.8\%$ to $31.2\%$ additional survival rate compared to the benchmark, for $\gamma \in (0.3676, 0.4059)$. Similarly, for the $F_{th} = 0.9$ case we observe a survival rate addition between $30.6\%$ and $31.2\%$, for $\gamma \in (0.1056, 0.1085)$. 

The limited range of $\gamma$ values is easily interpreted, as the values lower than this produce states with fidelity above the threshold with no filtering necessary, thus the optimal choice is to not use Gisin's local filter. In other words, the channel introduces such an insignificant amount of noise that the entanglement can simply pass through the channel without any filtering and still maintain high fidelity. For $\gamma$ values above this range, the amplitude damping effect is so strong that there does not exist a $\{\alpha, \beta\}$ values such that filtering will achieve a fidelity greater than $F_{th}$. To achieve high-fidelity entanglement in the high $\gamma$ range, one could cascade filters in series with $\texttt{Filter}_\texttt{{*,1}}$ which constitutes an orthogonal research topic. 

An analysis of the contributions of different filtering events to the final survival rate reveals that the improvement comes mostly from the $\tilde{\rho}_{10}$ and $\tilde{\rho}_{01}$ cases --- implying that one photon is reflected in one arm while its entangled counterpart is transmitted in the other arm, as shown in the histogram in Fig. \ref{full_histogram}. This indicates that our proposed qubit recycling protocol effectively recycles the disposed photons in these cases, leading to a higher overall survival rate without compromising the fidelity of the entangled photon pairs.

\subsection{Partial Filter Results}

Our simulation results for the partial filtering scheme show an improvement in the survival rate of similar degree to the full filtering scheme, which can be seen in Fig. \ref{partial_0.7} and Fig. \ref{partial_0.9}. Specifically, for the $F_{th} = 0.7$ case, the partial filtering scheme adds between $20.5\%$ and $25.0\%$ to the benchmark survival rate, for $\gamma \in (0.3676, 0.3824)$. For the $F_{th} = 0.9$ case, we similarly observe an additional $24.3\%$-$25.0\%$ increase in survival rate, for $\gamma \in (0.1056, 0.1079)$. 

We note an observed tradeoff between the full and partial filtering schemes. The partial filtering scheme has $\gamma$ ranges for which it is a viable design which are proper subsets of the corresponding full filter's $\gamma$ ranges, however the overall survival rates are significantly higher for the partial filtering scheme within those ranges. For the $F_{th} = 0.7$ case, the highest survival rate using the full filtering scheme is $56.1\%$, while the corresponding partial filtering scheme has a $74.9\%$ survival rate. A similar difference is observed in the $F_{th} = 0.9$ case, where the full filtering scheme has a maximum survival rate of $56.2\%$, and the corresponding partial filtering scheme has a survival rate of $75.0\%$.

This tradeoff can be explained by the increase in the probability of photons being initially transmitted through the first filter. Given that the filter is on only one side of the entanglement pair in the partial filtering scheme, compared to both sides in the full filtering scheme, the probability of being transmitted is much greater. This allows for less filtering being possible, though, which explains the smaller $\gamma$ ranges for which we see a gain in survival rate. These differences in contribution of the transmitted photons can be seen by comparing the histograms Fig. \ref{full_histogram} and Fig. \ref{partial_histogram}.

\subsection{Synchronization and Multi-party Agreement}

The results confirm the effectiveness of our qubit recycling protocol in enhancing the performance of entanglement distillation in both the full filtering and partial filtering schemes. However, both the full filtering scheme as well as the partial filtering scheme suffer from a potential synchronization challenge, which occurs in the $\tilde{\rho}_{10}$ and $\tilde{\rho}_{01}$ cases in the full filtering scheme, or the $\Tilde{\rho}_{0}$ case in the partial filtering scheme, where one photon in an entangled pair passes through its first filter (or is not filtered in the case of the partial filtering scheme), while the corresponding photon reflects off of its respective first filter, and subsequently passes through its second filter. As a result, the arrival times of the photons at Alice's and Bob's detectors will be different, leading to a discrepancy in their timesheets. When Alice and Bob compare their timesheets to identify photon coincidences, this discrepancy may cause difficulties in recognizing these events as coincidences, potentially leading them to be incorrectly discarded.

This time discrepancy can be avoided if Alice and Bob each measure the distance of their respective recycled light paths and share this information with each other, as well as the entanglement source. Alice and Bob can then compensate for the time difference for the recycled photons. In addition, the entanglement source can also use this information to emit photons only at intervals which are not equal to the interval between the arrivals of the entangled photons in these cases. This allows Alice and Bob to be certain that any photons arriving with such an interval between them can in fact be labeled a coincidence pair.

Furthermore, it is important to note that in the full filtering scheme, even if the $\tilde{\rho}_{10}$ and $\tilde{\rho}_{01}$ cases are excluded, the inclusion of the $\tilde{\rho}_{00}$ case alone still results in a  benefit in survival rate, albeit at a lower amount. Specifically, we see a $6.06\%$ increase in survival rate for $F_{th} = 0.7$, and a $6.24\%$ increase for $F_{th} = 0.9$. This is illustrated by Fig. \ref{full_histogram}. 
\section{Conclusion and Future Work}\label{conclusion}
In this paper, we have presented a novel qubit recycling protocol for improving the yield of high-fidelity entangled qubits in photonic quantum systems. By employing a second local filter, our approach effectively reclaims discarded entangled qubits, resulting in a substantial increase in the yield of high-fidelity entanglement pairs. Our proposed protocol achieves up to a 31.2\% gain in high-fidelity survival rate while incurring only moderate system complexity in terms of invested hardware and extra signaling for synchronization. Our work demonstrates the potential of qubit recycling in quantum entanglement distillation, which could have implications for the development of scalable and robust quantum communication networks.

An avenue for future work is to examine the applications of qubit recycling in different network models (e.g. multipartite entanglement, non-symmetric noise channels.) Another avenue is examining the local filter with a zero-valued parameter which breaks entanglement in the reflected photons. In some network models, using such a filter can be optimal, so finding use of these photons could lead to improvement over our proposed protocol.

\bibliography{ref}
\bibliographystyle{IEEEtran}

\end{document}